\documentclass[journal]{IEEEtran}

\usepackage{amsmath,amssymb,amsthm,mathtools}
\usepackage{graphicx}
\usepackage{booktabs}
\usepackage{array}
\usepackage{tabularx}
\usepackage[table]{xcolor}
\usepackage{cite}
\usepackage{url}
\usepackage[hidelinks]{hyperref}
\usepackage{xfp} 
\usepackage{tikz}

\newtheorem{theorem}{Theorem}
\newtheorem{lemma}{Lemma}
\newtheorem{proposition}{Proposition}
\newtheorem{corollary}{Corollary}
\newtheorem{claim}{Claim}

\theoremstyle{definition}
\newtheorem{definition}{Definition}

\theoremstyle{remark}

\newcommand{\Nstar}{N^{\star}}
\newcommand{\Shift}{\mathsf{SHIFT}}
\newcommand{\Cover}{\mathsf{COVER}}
\newcommand{\GASP}{\mathsf{GASP}}
\newcommand{\DOG}{\mathsf{DOG}}

\newcommand{\usefulentry}[1]{\cellcolor{red!25}#1}
\newcommand{\maskcolumn}[1]{\cellcolor{blue!25}#1}
\newcommand{\maskrow}[1]{\cellcolor{green!25}#1}

\newcommand{\motivatingscale}{1.10}

\newcommand{\largetablescale}{\fpeval{0.6552*\motivatingscale}}

\title{Improved Degree Tables for \\ Secure Distributed Matrix Multiplication}

\author{John Byrne, Rafael G. L. D'Oliveira, and Michael Tait\thanks{The third author was partially supported by NSF grant DMS 2245556, a Villanova University Summer Grant, and a Villanova CLAS Research Semester.

  } \\ Department of Mathematics, University of California San Diego, USA, jobyrne@ucsd.edu \\ School of Mathematical and Statistical Sciences, Clemson University, USA, rdolive@clemson.edu \\ Department of Mathematics \& Statistics, Villanova University, USA, michael.tait@villanova.edu
}

\begin{document}

\maketitle

\begin{abstract}
In secure distributed matrix multiplication, a user wishes to compute the product of two matrices with the assistance of servers, in such a way that any $T$ colluding servers learn nothing about either matrix. Degree tables are a combinatorial tool for constructing polynomial codes for this problem and underlie several state-of-the-art schemes, including $\GASP_r$, $\GASP_{r,s}$, and $\DOG_{r,s}$. We introduce a periodic-gap framework for constructing degree tables that contains these three families as special cases and leads to two new constructions, $\Shift_{r,s}$ and $\Cover_r$. We determine their exact recovery thresholds and show that, in many cases, they outperform the current state of the art. We also prove a new lower bound on the recovery threshold of ordinary integer degree tables. In the balanced case, in which the partitioning parameters and the security parameter are equal, we sharpen this bound to match $\Shift_{r,s}$ up to lower-order terms, showing that it is asymptotically optimal among degree-table constructions.
\end{abstract}

 

\section{Introduction}
\label{sec:introduction}

We consider the problem of secure distributed matrix multiplication (SDMM). A user has two matrices $A$ and $B$ and wishes to compute $AB$ with the assistance of several servers. The servers perform the requested computations correctly, but any set of up to $T$ servers may collude in an attempt to learn information about $A$ or $B$. Polynomial codes are among the best-performing known approaches to this problem.

We use an outer-product partitioning (OPP), writing $A^\intercal=\begin{bmatrix}A_1^\intercal&\cdots&A_K^\intercal\end{bmatrix}$ and $B=\begin{bmatrix}B_1&\cdots&B_L\end{bmatrix}$, so that $AB$ is the $K\times L$ block matrix with entries $A_kB_\ell$. Computing $AB$ is therefore equivalent to computing all $KL$ subproducts $A_kB_\ell$. To do this securely, we define \begin{align*} 
f(x)&=\sum_{k=1}^{K}A_kx^{a_k}+\sum_{t=1}^{T}R_tx^{c_t},\\ g(x)&=\sum_{\ell=1}^{L}B_\ell x^{b_\ell}+\sum_{t=1}^{T}S_tx^{d_t}, 
\end{align*}
where $R_t$ and $S_t$ are independently and uniformly distributed random matrices. The user sends evaluations of $f$ and $g$ to each server, and the server returns the corresponding evaluation of $h(x)=f(x)g(x)$. From sufficiently many responses, the user recovers the relevant coefficients of $h(x)$ and obtains $A_kB_\ell$ from the coefficient of $x^{a_k+b_\ell}$.

\begin{figure}[!t]
\centering
\setlength{\tabcolsep}{3pt}
\resizebox{\columnwidth}{!}{%
\begin{tabular}{c|ccc|ccc}
 & $b_1$ & $\!\!\cdots\!\!$ & $b_L$ & \cellcolor{blue!25} $d_1$ & \cellcolor{blue!25} $\!\!\cdots\!\!$ & \cellcolor{blue!25} $d_T$ \\
\toprule
$a_1$ & \cellcolor{red!25} $a_1+b_1$ & \cellcolor{red!25} $\cdots$ & \cellcolor{red!25} $a_1+b_L$ & $a_1+d_1$ & $\cdots$ & $a_1+d_T$ \\
$\vdots$ & \cellcolor{red!25} $\vdots$ & \cellcolor{red!25} $\ddots$ & \cellcolor{red!25} $\vdots$ & $\vdots$ & $\ddots$ & $\vdots$ \\
$a_K$ & \cellcolor{red!25} $a_K+b_1$ & \cellcolor{red!25} $\cdots$ & \cellcolor{red!25} $a_K+b_L$ & $a_K+d_1$ & $\cdots$ & $a_K+d_T$ \\
\midrule
\cellcolor{green!25} $c_1$ & $c_1+b_1$ & $\cdots$ & $c_1+b_L$ & $c_1+d_1$ & $\cdots$ & $c_1+d_T$ \\
\cellcolor{green!25} $\vdots$ & $\vdots$ & $\ddots$ & $\vdots$ & $\vdots$ & $\ddots$ & $\vdots$ \\
\cellcolor{green!25} $c_T$ & $c_T+b_1$ & $\cdots$ & $c_T+b_L$ & $c_T+d_1$ & $\cdots$ & $c_T+d_T$ \\
\bottomrule
\end{tabular}%
}
\caption{The degree table for $f(x)g(x)$. The exponents $a_k$ and $b_\ell$ correspond to data blocks, while $c_t$ and $d_t$ correspond to random masks. Each entry is the sum of its row and column exponents and hence a monomial degree in $f(x)g(x)$. The red block corresponds to the desired products, while the other blocks contain interference. Decodability requires every red entry to be unique in the entire table, while security requires the green and blue exponents to be pairwise distinct within their respective sets. The goal is to select exponents that satisfy these constraints while minimizing the number of distinct entries.}
\label{fig:degree-table}
\end{figure}

The number of servers needed to recover $AB$ is called the recovery threshold. For polynomial codes of the above form, the recovery threshold is governed by the choice of exponents and can be studied through a degree table~\cite{9004505}, illustrated in Fig.~\ref{fig:degree-table}. The degree-table viewpoint underlies many of the state-of-the-art constructions for the outer-product partitioning. The main constructions based on degree tables include $\GASP_r$, $\GASP_{r,s}$, and $\DOG_{r,s}$~\cite{9004505,9508383,11195364}. Related extensions include cyclic-addition degree tables for root-of-unity constructions such as CAT~\cite{9965858,11195364} and pole number tables for PoleGap~\cite{10858081}.


\newcommand{\motivatingpanelwidth}{0.24\textwidth}

\newcommand{\firsttablescale}{\fpeval{0.80*\motivatingscale}}

\newcommand{\motivatingcolsep}{1.2pt}

\newcommand{\motivatingrowstretch}{0.90}

\newcommand{\motivatinglabelgap}{4.4em}

\newcommand{\motivatingcaptiongap}{0.9em}

\begin{figure*}[!t]
\centering
\setlength{\tabcolsep}{\motivatingcolsep}
\renewcommand{\arraystretch}{\motivatingrowstretch}

\begin{tabular}{@{}c@{}c@{}c@{}c@{}}

\makebox[\motivatingpanelwidth][c]{%
\scalebox{\firsttablescale}{%
\begin{tabular}{c|ccc|cccc}
 &0&7&14&
 \maskcolumn{19}&\maskcolumn{20}&
 \maskcolumn{21}&\maskcolumn{22}\\
\toprule
0&
 \usefulentry{0}&\usefulentry{7}&\usefulentry{14}&
 19&20&21&22\\
1&
 \usefulentry{1}&\usefulentry{8}&\usefulentry{15}&
 20&21&22&23\\
2&
 \usefulentry{2}&\usefulentry{9}&\usefulentry{16}&
 21&22&23&24\\
3&
 \usefulentry{3}&\usefulentry{10}&\usefulentry{17}&
 22&23&24&25\\
4&
 \usefulentry{4}&\usefulentry{11}&\usefulentry{18}&
 23&24&25&26\\
\midrule
\maskrow{5}&
 5&12&19&
 24&25&26&27\\
\maskrow{6}&
 6&13&20&
 25&26&27&28\\
\maskrow{12}&
 12&19&26&
 31&32&33&34\\
\maskrow{13}&
 13&20&27&
 32&33&34&35\\
\bottomrule
\end{tabular}%
}}%
&
\makebox[\motivatingpanelwidth][c]{%
\scalebox{\firsttablescale}{%
\begin{tabular}{c|ccc|cccc}
 &10&17&24&
 \maskcolumn{0}&\maskcolumn{1}&
 \maskcolumn{2}&\maskcolumn{3}\\
\toprule
2&
 \usefulentry{12}&\usefulentry{19}&\usefulentry{26}&
 2&3&4&5\\
3&
 \usefulentry{13}&\usefulentry{20}&\usefulentry{27}&
 3&4&5&6\\
4&
 \usefulentry{14}&\usefulentry{21}&\usefulentry{28}&
 4&5&6&7\\
5&
 \usefulentry{15}&\usefulentry{22}&\usefulentry{29}&
 5&6&7&8\\
6&
 \usefulentry{16}&\usefulentry{23}&\usefulentry{30}&
 6&7&8&9\\
\midrule
\maskrow{0}&
 10&17&24&
 0&1&2&3\\
\maskrow{1}&
 11&18&25&
 1&2&3&4\\
\maskrow{7}&
 17&24&31&
 7&8&9&10\\
\maskrow{8}&
 18&25&32&
 8&9&10&11\\
\bottomrule
\end{tabular}%
}}%
&
\makebox[\motivatingpanelwidth][c]{%
\scalebox{\largetablescale}{%
\begin{tabular}{c|ccccc|cccccc}
 &0&5&10&15&20&
 \maskcolumn{25}&\maskcolumn{26}&\maskcolumn{27}&
 \maskcolumn{28}&\maskcolumn{29}&\maskcolumn{30}\\
\toprule
0&
 \usefulentry{0}&\usefulentry{5}&\usefulentry{10}&
 \usefulentry{15}&\usefulentry{20}&
 25&26&27&28&29&30\\
1&
 \usefulentry{1}&\usefulentry{6}&\usefulentry{11}&
 \usefulentry{16}&\usefulentry{21}&
 26&27&28&29&30&31\\
2&
 \usefulentry{2}&\usefulentry{7}&\usefulentry{12}&
 \usefulentry{17}&\usefulentry{22}&
 27&28&29&30&31&32\\
3&
 \usefulentry{3}&\usefulentry{8}&\usefulentry{13}&
 \usefulentry{18}&\usefulentry{23}&
 28&29&30&31&32&33\\
4&
 \usefulentry{4}&\usefulentry{9}&\usefulentry{14}&
 \usefulentry{19}&\usefulentry{24}&
 29&30&31&32&33&34\\
\midrule
\maskrow{25}&
 25&30&35&40&45&
 50&51&52&53&54&55\\
\maskrow{26}&
 26&31&36&41&46&
 51&52&53&54&55&56\\
\maskrow{27}&
 27&32&37&42&47&
 52&53&54&55&56&57\\
\maskrow{30}&
 30&35&40&45&50&
 55&56&57&58&59&60\\
\maskrow{31}&
 31&36&41&46&51&
 56&57&58&59&60&61\\
\maskrow{32}&
 32&37&42&47&52&
 57&58&59&60&61&62\\
\bottomrule
\end{tabular}%
}}%
&
\makebox[\motivatingpanelwidth][c]{%
\scalebox{\largetablescale}{%
\begin{tabular}{c|ccccc|cccccc}
 &0&7&14&21&28&
 \maskcolumn{35}&\maskcolumn{36}&\maskcolumn{37}&
 \maskcolumn{38}&\maskcolumn{39}&\maskcolumn{40}\\
\toprule
2&
 \usefulentry{2}&\usefulentry{9}&\usefulentry{16}&
 \usefulentry{23}&\usefulentry{30}&
 37&38&39&40&41&42\\
3&
 \usefulentry{3}&\usefulentry{10}&\usefulentry{17}&
 \usefulentry{24}&\usefulentry{31}&
 38&39&40&41&42&43\\
4&
 \usefulentry{4}&\usefulentry{11}&\usefulentry{18}&
 \usefulentry{25}&\usefulentry{32}&
 39&40&41&42&43&44\\
5&
 \usefulentry{5}&\usefulentry{12}&\usefulentry{19}&
 \usefulentry{26}&\usefulentry{33}&
 40&41&42&43&44&45\\
6&
 \usefulentry{6}&\usefulentry{13}&\usefulentry{20}&
 \usefulentry{27}&\usefulentry{34}&
 41&42&43&44&45&46\\
\midrule
\maskrow{0}&
 0&7&14&21&28&
 35&36&37&38&39&40\\
\maskrow{1}&
 1&8&15&22&29&
 36&37&38&39&40&41\\
\maskrow{7}&
 7&14&21&28&35&
 42&43&44&45&46&47\\
\maskrow{8}&
 8&15&22&29&36&
 43&44&45&46&47&48\\
\maskrow{14}&
 14&21&28&35&42&
 49&50&51&52&53&54\\
\maskrow{15}&
 15&22&29&36&43&
 50&51&52&53&54&55\\
\bottomrule
\end{tabular}%
}}%

\\[\motivatinglabelgap]

\makebox[\motivatingpanelwidth][c]{%
{\footnotesize (a) $\DOG_{2,4}$, $N=34$}}
&
\makebox[\motivatingpanelwidth][c]{%
{\footnotesize (b) $\Shift_{2,4}$, $N=33$}}
&
\makebox[\motivatingpanelwidth][c]{%
{\footnotesize (c) $\GASP_3$, $N=57$}}
&
\makebox[\motivatingpanelwidth][c]{%
{\footnotesize (d) $\Cover_2$, $N=56$}}

\end{tabular}

\vspace{\motivatingcaptiongap}

\caption{Two constructions that improve the state of the art. Panels~(a) and~(b) show $\DOG_{2,4}$~\cite{11195364} and $\Shift_{2,4}$, respectively, for $K=5$, $L=3$, and $T=4$; $\Shift_{2,4}$ lowers the recovery threshold from $34$ to $33$. Panels~(c) and~(d) show $\GASP_3$~\cite{9004505,9508383} and $\Cover_2$, respectively, for $K=L=5$ and $T=6$; $\Cover_2$ lowers the threshold from $57$ to $56$. The red block corresponds to the desired products, and the green and blue headers give the masking exponents.}
\label{fig:motivating-examples}

\end{figure*}

A degree table is determined by four exponent sets. Let $\mathcal A=\{a_1,\ldots,a_K\}$, $\mathcal B=\{b_1,\ldots,b_L\}$, $\mathcal C=\{c_1,\ldots,c_T\}$, and $\mathcal D=\{d_1,\ldots,d_T\}$. As shown in Fig.~\ref{fig:degree-table}, the degree table contains all sums in $(\mathcal A\cup\mathcal C)+(\mathcal B\cup\mathcal D)$. The block $\mathcal A+\mathcal B$ is the useful block, while the other three blocks contain interference terms. Decodability requires every useful entry to occur exactly once in the entire table, while degree-table security requires the elements of $\mathcal C$ and $\mathcal D$ to be pairwise distinct. Over a sufficiently large field, these conditions give a decodable and information-theoretically $T$-secure code whose recovery threshold is the number of distinct entries in the table. Thus, the goal is to choose the four exponent sets so as to minimize this number.

In this paper, we introduce a periodic-gap framework that contains $\GASP_r$, $\GASP_{r,s}$, and $\DOG_{r,s}$ as special cases and leads to two new constructions, $\Shift_{r,s}$ and $\Cover_r$. We derive their recovery thresholds, identify regimes in which they improve the previous constructions, prove a general lower bound for degree tables, and determine the optimum asymptotically in the balanced setting.

\subsection{Related Work}
\label{sec:related_work}

Polynomial codes were introduced for nonsecure distributed matrix multiplication in~\cite{yu2017polynomial} and followed by several extensions~\cite{8006963,8437871,8765375,8949560,yu2019lagrange,9519610,10786350}. Early information-theoretic formulations of SDMM appeared in~\cite{8647313,8382305}. Later work studied straggler tolerance~\cite{8675905}, flexible communication~\cite{8985291}, upload and download costs~\cite{8989342,9440909}, and total computation time~\cite{9162296}.

Degree tables were introduced in~\cite{9004505}. Subsequent work proved lower bounds and developed state-of-the-art constructions with integer exponents~\cite{9508383,11195364}. Beyond this setting, root-of-unity constructions use degree tables modulo an integer~\cite{9965858,11195364}, while PoleGap uses pole-number tables~\cite{10858081}.

Other approaches to outer-product SDMM include A3S~\cite{8675905} and secure bivariate polynomial coding~\cite{9681059}. A general framework based on linear codes and star products unifies many SDMM schemes and handles stragglers and Byzantine servers~\cite{10415397}. Other work studies inner-product partitioning~\cite{9732990,9965839,10206764}, general grid partitioning~\cite{9174167,9229375}, private and secure matrix multiplication with replicated or MDS-coded servers~\cite{9696353}, batch and multi-party matrix multiplication~\cite{9539194,9523544}, adaptive multi-message computation~\cite{9681896}, server cooperation~\cite{9681812}, and reduced downloads from each server~\cite{9606447}.

Degree tables have been extended to secure Gram matrix multiplication through symmetric degree tables~\cite{10161614} and to general grid partitioning~\cite{11653890}. Degree-table methods have also been used for straggler and adversary tolerance~\cite{e25020266,10478018}, precomputation~\cite{10619695}, private information retrieval~\cite{9929410}, and quantum secure distributed matrix multiplication~\cite{nomeir2025quantum,11462243}.

\subsection{Main Contributions}
\label{sec:main-contributions}

Our main contributions are as follows.

\begin{itemize}
\item We introduce a general periodic-gap framework for degree tables that contains $\GASP_r$, $\GASP_{r,s}$, and $\DOG_{r,s}$ as special cases.

\item We introduce two new constructions, $\Shift_{r,s}$ and $\Cover_r$, and derive their recovery thresholds. Together, these constructions improve the best previously known recovery thresholds for ordinary integer degree tables in explicit parameter regimes. In particular, when $T=ur=vs$, we prove $N_{\DOG_{r,s}}-N_{\Shift_{r,s}}=s-r-1$, so $\Shift_{r,s}$ strictly improves $\DOG_{r,s}$ whenever $s\geq r+2$. We also identify regimes in which $\Cover_r$ improves $\GASP_{\mathrm{big}}$, including parameters for which it achieves a new best-known recovery threshold.

\item Writing $\Nstar(K,L,T)$ for the minimum recovery threshold among ordinary integer degree tables, we prove the general lower bound $\Nstar(K,L,T)\geq KL+(KLT^2)^{1/3}$. 

\item In the balanced setting $K=L=T=n$, we prove $\Nstar(n,n,n)=n^2+3n^{4/3}+O(n^{7/6})$. This matches the recovery threshold of $\Shift_{r,s}$ up to lower-order terms, proving that it is asymptotically optimal.
\end{itemize}

\section{Two Motivating Examples}
\label{sec:motivating-examples}

We illustrate the two mechanisms with two small examples in which the new constructions improve the best-known recovery thresholds for OPP SDMM. Fig.~\ref{fig:motivating-examples} compares their degree tables with those of the previous best constructions.

\subsection{The $\Shift_{2,4}$ Construction for $K=5$, $L=3$, and $T=4$}

For partitioning parameters $K=5$ and $L=3$ and security parameter $T=4$, the best previously known recovery threshold for OPP SDMM is $34$, achieved by $\DOG_{2,4}$~\cite{11195364}. More precisely,
\begin{align*}
\min_r N_{\GASP_r}&=35
&&\text{at }r=2,3,\text{ or }4,\\
\min_{r,s}N_{\GASP_{r,s}}&=35
&&\text{at }(r,s)=(2,4),(3,4),\text{ or }(4,4),\\
\min_{r,s}N_{\DOG_{r,s}}&=34
&&\text{at }(r,s)=(2,4).
\end{align*}
The published $\mathrm{CAT}_x$ and PoleGap constructions are not applicable to these parameters \cite{11195364}. The $\DOG_{2,4}$ construction, shown in Fig.~\ref{fig:motivating-examples}(a), uses $\mathcal A=\{0,1,2,3,4\}$, $\mathcal B=\{0,7,14\}$, $\mathcal C=\{5,6,12,13\}$, and $\mathcal D=\{19,20,21,22\}$. The $\Shift_{2,4}$ construction uses $\mathcal A=\{2,3,4,5,6\}$, $\mathcal B=\{10,17,24\}$, $\mathcal C=\{0,1,7,8\}$, and $\mathcal D=\{0,1,2,3\}$. The resulting degree table is shown in Fig.~\ref{fig:motivating-examples}(b).

The improvement comes from increasing the overlap among the interference blocks. In particular, $\mathcal A+\mathcal D=\{2,3,\ldots,9\}$ and $\mathcal C+\mathcal D=\{0,1,\ldots,4\}\cup\{7,8,\ldots,11\}$. These two blocks overlap in the six values $\{2,3,4,7,8,9\}$. The remaining interference block $\mathcal C+\mathcal B$ overlaps their union at the values $10$ and $11$. The blocks $\mathcal A+\mathcal D$, $\mathcal C+\mathcal D$, and $\mathcal C+\mathcal B$ have sizes $8$, $10$, and $8$, respectively, and their union has $18$ distinct values. Therefore, the recovery threshold is $N=15+18=33$, one fewer server than the previous best value $34$.

\subsection{The $\Cover_2$ Construction for $K=L=5$ and $T=6$}

For partitioning parameters $K=L=5$ and security parameter $T=6$, the best previously known recovery threshold for OPP SDMM is $57$~\cite{9508383,11195364}. More precisely,
\begin{align*}
\min_r N_{\GASP_r}&=57
&&\text{at }r=3,\\
\min_{r,s}N_{\GASP_{r,s}}&=57
&&\text{at }(r,s)=(2,3)\text{ or }(3,5),\\
\min_{r,s}N_{\DOG_{r,s}}&=57
&&\text{at }(r,s)=(2,3).
\end{align*}
The $\mathrm{CAT}_x$ and PoleGap constructions are not applicable to these parameters \cite{11195364}. The $\GASP_3$ construction, shown in Fig.~\ref{fig:motivating-examples}(c), uses $\mathcal A=\{0,1,2,3,4\}$, $\mathcal B=\{0,5,10,15,20\}$, $\mathcal C=\{25,26,27,30,31,32\}$, and $\mathcal D=\{25,26,27,28,29,30\}$. The $\Cover_2$ construction uses $\mathcal A=\{2,3,4,5,6\}$, $\mathcal B=\{0,7,14,21,28\}$, $\mathcal C=\{0,1,7,8,14,15\}$, and $\mathcal D=\{35,36,37,38,39,40\}$. The resulting degree table is shown in Fig.~\ref{fig:motivating-examples}(d).

The improvement comes from making one interference block entirely redundant. In particular, $\mathcal A+\mathcal D=\{37,38,\ldots,46\}\subseteq\mathcal C+\mathcal D=\{35,36,\ldots,55\}$, so $\mathcal A+\mathcal D$ contributes no new degrees. The remaining interference block $\mathcal C+\mathcal B$ contains $14$ distinct values and overlaps $\mathcal C+\mathcal D$ in the four values $\{35,36,42,43\}$. Since $\mathcal C+\mathcal D$ has $21$ distinct values, the three interference blocks have $31$ distinct values in total. Therefore, the recovery threshold is $N=25+31=56$, one fewer server than the previous best value $57$.

\section{Main Results}
\label{sec:main-results}

We now state the main constructions and bounds. Their proofs are given in the following sections. Throughout, we assume $L\le K$, since the case $K<L$ follows from $AB=(B^\top A^\top)^\top$. For integers $a$ and $b$, write $[a,b)=\{a,a+1,\ldots,b-1\}$ if $a<b$, and $[a,b)=\varnothing$ otherwise.

\subsection{The Periodic-Gap Framework}

We introduce a periodic-gap framework for building degree tables. Its exponent sets are all periodic-gap sets: intervals, or more generally sets formed by repeating an interval at regular spacing along the integers. The framework unifies the principal constructions in the literature based on ordinary integer degree tables, with $\GASP_r$, its two-parameter extension $\GASP_{r,s}$, and $\DOG_{r,s}$ arising as special cases. Our new constructions, $\Shift_{r,s}$ and $\Cover_r$, use the same periodic-gap sets.

\begin{definition}
\label{def:periodic-gap-set}
For positive integers $T$, $P$, and $w$, we denote by $\operatorname{gap}(T,P,w)$ the set of $T$ smallest elements of $P\mathbb Z_{\ge 0}+[0,w)$. We call it a periodic-gap set, with period $P$ and width $w$.
\end{definition}

In other words, we begin with the interval $[0,w)$, repeat it every $P$ positions, and stop after selecting $T$ elements. For example, $\operatorname{gap}(5,7,2)=\{0,1,7,8,14\}$. The first two elements form the block $\{0,1\}$, the next two form the translated block $\{7,8\}$, and the final element begins the next block. Whenever $w\ge T$, all $T$ elements fit in the first block, so $\operatorname{gap}(T,P,w)=[0,T)$. Thus, an interval is a special case of a periodic-gap set.

\begin{definition}
\label{def:periodic-gap-framework}
A degree table belongs to the periodic-gap framework if its exponent sets can be written as
\begin{align*}
\mathcal A&=x_A+[0,K),&
\mathcal B&=x_B+P[0,L),\\
\mathcal C&=x_C+\operatorname{gap}(T,P,r),&
\mathcal D&=x_D+\operatorname{gap}(T,P,s),
\end{align*}
for some positive integers $P,r,s$ and integer shifts $x_A,x_B,x_C,x_D$, where $P[0,L)=\{0,P,2P,\ldots,(L-1)P\}$.
\end{definition}

The parameters $P$, $r$, and $s$ determine the shapes of the exponent sets, while the shifts $x_A$, $x_B$, $x_C$, and $x_D$ determine how these sets are placed relative to one another.

\begin{proposition}
\label{prop:periodic-gap-unification}
The families $\GASP_r$, $\GASP_{r,s}$, and $\DOG_{r,s}$ belong to the periodic-gap framework.
\end{proposition}

Indeed, it is enough to choose the period, widths, and shifts as follows. The family $\GASP_r$, for $1\le r\le\min\{K,T\}$, uses period $P=K$ and
\begin{align*}
\mathcal A&=[0,K),&
\mathcal B&=K[0,L),\\
\mathcal C&=KL+\operatorname{gap}(T,K,r),&
\mathcal D&=KL+[0,T).
\end{align*}
Thus, $\mathcal C$ is a periodic-gap set of width $r$, while $\mathcal D$ is an interval. The choice $r=\min\{K,T\}$ is denoted by $\GASP_{\mathrm{big}}$.

The two-parameter family $\GASP_{r,s}$, defined for $K,L,T\ge2$ and $1\le r,s\le\min\{K,T\}$, keeps the same period and placement, but replaces $\mathcal D$ by a second periodic-gap set:
\begin{align*}
\mathcal A&=[0,K),&
\mathcal B&=K[0,L),\\
\mathcal C&=KL+\operatorname{gap}(T,K,r),&
\mathcal D&=KL+\operatorname{gap}(T,K,s).
\end{align*}

The family $\DOG_{r,s}$, defined for $K,L,T\ge2$, $1\le r\le T$, and $1\le s\le\min\{T,K+r\}$, uses the larger period $P=K+r$ and a different placement:
\begin{align*}
\mathcal A&=[0,K),&
\mathcal B&=P[0,L),\\
\mathcal C&=K+\operatorname{gap}(T,P,r),&
\mathcal D&=P(L-1)+K+\operatorname{gap}(T,P,s).
\end{align*}

The two constructions introduced in this paper are new within the same framework: $\Shift_{r,s}$ is designed to increase the overlap among the three interference blocks, while $\Cover_r$ places one interference block entirely inside another.

\subsection{The \texorpdfstring{$\Shift_{r,s}$}{SHIFT(r,s)} Construction}

We now use the periodic-gap framework to define the first new family, $\Shift_{r,s}$. Like $\DOG_{r,s}$, the construction uses periodic-gap sets of widths $r$ and $s$. The difference is in their placement: $\Shift_{r,s}$ shifts the four exponent sets with the aim of increasing the overlap among the three interference blocks.

Let $r$ and $s$ satisfy $1\le r<s\le \min\{T,K+1\}$, and set $P=K+r$. Let $u=\lceil T/r\rceil$ and $v=\lceil T/s\rceil$ be the numbers of blocks in the two periodic-gap sets, and let $\rho=T-(u-1)r$ and $\sigma=T-(v-1)s$ be the sizes of their final blocks, so that $1\le\rho\le r$ and $1\le\sigma\le s$.

\begin{definition}
\label{def:shift-main-results}
The $\Shift_{r,s}$ construction uses exponents
\begin{align*}
\mathcal A&=(u-2)P+r+[0,K),\\
\mathcal B&=vP+\rho+\sigma-r-1+P[0,L),\\
\mathcal C&=\operatorname{gap}(T,P,r),\\
\mathcal D&=\operatorname{gap}(T,P,s).
\end{align*}
\end{definition}

\begin{proposition}
\label{prop:shift-decodable}
The $\Shift_{r,s}$ construction is decodable and $T$-secure.
\end{proposition}

In the following theorem, we determine the recovery threshold of $\Shift_{r,s}$. For this, write $x_+=\max\{x,0\}$ for a real number $x$.

\begin{theorem}
\label{thm:shift-count-main-results}
The recovery threshold of $\Shift_{r,s}$ is
\begin{multline*}
N_{\Shift_{r,s}}
=
KL+(u+v-1)(r+s-1)+rL\\
{}+v(K-s+1)-2(r-\rho)-(s-\sigma)\\
{}+(u-v-1)_+
\min\bigl\{K-s+1,(r-\rho-\sigma+1)_+\bigr\}.
\end{multline*}
\end{theorem}

When both $r$ and $s$ divide $T$, we have $\rho=r$ and $\sigma=s$, and the formula in Theorem~\ref{thm:shift-count-main-results} simplifies. In this case, we show that $\Shift_{r,s}$ always matches or improves on $\DOG_{r,s}$.

\begin{corollary}
\label{cor:shift-dog-main-results}
Suppose that $K$, $L$, and $T$ are greater or equal to $2$, that $1\le r<s\le\min\{T,K+1\}$, and $T=ur=vs$. Then,
\begin{align*}
N_{\Shift_{r,s}}
={}&KL+(u+v-1)(r+s-1)\\
&+rL+v(K-s+1).
\end{align*}
For the corresponding $\DOG_{r,s}$ construction,
\[
N_{\DOG_{r,s}}-N_{\Shift_{r,s}}=s-r-1.
\]
Therefore, $\Shift_{r,s}$ matches $\DOG_{r,s}$ when $s=r+1$ and strictly improves it when $s\ge r+2$.
\end{corollary}

As shown in Fig.~\ref{fig:motivating-examples}(b), $\Shift_{2,4}$ reduces the recovery threshold from the previous best value in our comparison, $34$, to $33$ when the partitioning parameters are $K=5$ and $L=3$, and the security parameter is $T=4$.

\subsection{The \texorpdfstring{$\Cover_r$}{COVER(r)} Construction}

We now use the periodic-gap framework to define the second new family. The construction $\Cover_r$ places the entire interference block $\mathcal A+\mathcal D$ inside $\mathcal C+\mathcal D$, so that $\mathcal A+\mathcal D$ contributes no additional degrees. Assume $K<T$, and let $r$ satisfy $1\le r<T$. Set $P=K+r$.

\begin{definition}
\label{def:cover-main-results}
The $\Cover_r$ construction uses exponents
\begin{align*}
\mathcal A&=r+[0,K),\\
\mathcal B&=P[0,L),\\
\mathcal C&=\operatorname{gap}(T,P,r),\\
\mathcal D&=LP+[0,T).
\end{align*}
\end{definition}

\begin{proposition} \label{prop:cover-decodable}
    The $\Cover_r$ construction is decodable and $T$-secure.
\end{proposition}

We now determine the recovery threshold of $\Cover_r$. 

\begin{theorem}
\label{thm:cover-count-main-results}
The recovery threshold of $\Cover_r$ is
\begin{align*}
N_{\Cover_r}
={}&KL+Lr+\left(\left\lceil\frac{T}{r}\right\rceil-1\right)K+2T-1.
\end{align*}
\end{theorem}

The containment $\mathcal A+\mathcal D\subseteq\mathcal C+\mathcal D$ is what distinguishes $\Cover_r$: the upper-right interference block is completely redundant. The $\Cover_r$ construction lies in the regime where $L\le K<T$. In this regime the original GASP scheme uses $\GASP_{\mathrm{big}}$~\cite{9004505}. In this regime, $\GASP_{\mathrm{big}}=\GASP_K$ and has recovery threshold $2KL+2T-1$~\cite{9004505}. The following corollary gives a simple condition under which $\Cover_r$ improves it.

\begin{corollary}
\label{cor:cover-gasp-main-results}
If $Lr+\left(\left\lceil\frac{T}{r}\right\rceil-1\right)K<KL$,
then $\Cover_r$ has a strictly smaller recovery threshold than $\GASP_{\mathrm{big}}$.
\end{corollary}

For example, let $L=3$, $T=K+1$, and let $K\ge5$ be odd. Choosing $r=(K+1)/2$, $\Cover_r$ saves $(K-3)/2$ servers over $\GASP_{\mathrm{big}}$. For $K=5$, $L=3$, and $T=6$, $\GASP_{\mathrm{big}}$ attains the previous best-known recovery threshold $41$, while $\Cover_3$ achieves $40$. As shown in Fig.~\ref{fig:motivating-examples}(d), $\Cover_2$ also improves the state of the art for $K=L=5$ and $T=6$, lowering the best-known recovery threshold from $57$ to $56$.

\subsection{A General Lower Bound}

The following lower bound holds for all partitioning and security parameters.

\begin{theorem}
\label{thm:general-lower-bound-main-results}
For all positive integers $K$, $L$, and $T$,
\[
\Nstar(K,L,T)\ge KL+(KLT^2)^{1/3}.
\]
\end{theorem}

The term $KL$ counts the distinct useful entries, while the second term lower-bounds the number of distinct interference entries. In the balanced setting $K=L=T=n$, the theorem gives $\Nstar(n,n,n)\ge n^2+n^{4/3}$, whereas \cite[Theorem~2]{9508383} gives $n^2+3n$ for $n\geq2$. Thus, the new bound raises the asymptotic order of the excess above $n^2$ from $n$ to $n^{4/3}$. In the next subsection, we sharpen the leading coefficient from $1$ to $3$.

\subsection{Optimality in the Balanced Regime}

We now consider the balanced setting $K=L=T=n$. Every valid degree table contains $n^2$ useful degrees, so it remains to lower-bound the number of distinct degrees in the three interference blocks. The following combinatorial theorem gives the required bound for arbitrary subsets of $\mathbb R$ and may be of independent interest.

\begin{theorem}
\label{thm:balanced-lower-main-results}
Let $n$ be a positive integer, and $A,B,C,D\subseteq\mathbb R$ satisfy $|A|=|B|=|C|=|D|=n$ and $|A+B|=n^2$. Then
\[
\left|(A+D)\cup(C+B)\cup(C+D)\right|
\ge
3n^{4/3}-O(n^{7/6}),
\]
where the $O(\cdot)$ term is uniform over all choices of $A,B,C,D$. Consequently, $\Nstar(n,n,n)
\ge
n^2+3n^{4/3}-O(n^{7/6})$.
\end{theorem}

The $\Shift$ construction gives a matching upper bound up to lower-order terms.

\begin{corollary}
\label{cor:balanced-shift-upper-main-results}
As $n\to\infty$, $\Nstar(n,n,n)
\le
n^2+3n^{4/3}+O(n)$. 
In particular, when $n=m^3$ with $m\ge2$, choosing $r=m$ and $s=m^2$ in Theorem~\ref{thm:shift-count-main-results} gives
\[
N_{\Shift_{m,m^2}}
=
n^2+3n^{4/3}+n-n^{2/3}-n^{1/3}+1.
\]
\end{corollary}

Combining Theorem~\ref{thm:balanced-lower-main-results} and Corollary~\ref{cor:balanced-shift-upper-main-results} gives the following asymptotic result.

\begin{theorem}
\label{thm:balanced-main-results}
As $n\to\infty$, $\Nstar(n,n,n)
=
n^2+3n^{4/3}+O(n^{7/6})$.
\end{theorem}

Thus, the excess above the unavoidable $n^2$ useful degrees is $3n^{4/3}+O(n^{7/6})$, and the $\Shift$ construction attains the leading term. Any improvement can affect only lower-order terms. The interference bound requires only $|\mathcal A+\mathcal B|=|\mathcal A||\mathcal B|$, which says that the useful degrees are distinct within their block. Decodability then ensures that the useful and interference degrees are disjoint, yielding the lower bound on $\Nstar(n,n,n)$.

\section{Proofs of the main results}

In this section, we give proofs of the upper bounds in Theorems \ref{thm:shift-count-main-results} and \ref{thm:cover-count-main-results} and lower bounds in Theorems \ref{thm:general-lower-bound-main-results} and \ref{thm:balanced-lower-main-results}.

For the proofs, write $\alpha=(a_1,\ldots,a_K,c_1,\ldots,c_T)=(\alpha',\alpha'')$ and $\beta=(b_1,\ldots,b_L,d_1,\ldots,d_T)=(\beta',\beta'')$, where $\alpha'$ and $\beta'$ contain the data exponents and $\alpha''$ and $\beta''$ contain the masking exponents. We write $\alpha\oplus\beta$ for the $(K+T)\times(L+T)$ array whose $(i,j)$ entry is $\alpha_i+\beta_j$. We partition this array into the four blocks $\mathrm{UL}=\alpha'\oplus\beta'$, $\mathrm{UR}=\alpha'\oplus\beta''$, $\mathrm{LL}=\alpha''\oplus\beta'$, and $\mathrm{LR}=\alpha''\oplus\beta''$. For any subarray $M$, let $t(M)$ denote the set of distinct entries appearing in $M$.

\subsection{Upper bounds}

\begin{proof}[Proof of Proposition \ref{prop:shift-decodable}]
   It is obvious that $\Shift_{r,s}$ is $T$-secure. Thus, we consider decodability. If the scheme is not decodable, then some term $x$ of UL equals some $y$ in one of UL, UR, LL, or LR, where $y$ occupies a distinct position from $x$.

    \textit{Case 1:} $y$ appears in UL. Then we have $(u-2)P+r+x_1+vP+\rho+\sigma-r-1+Py_1=(u-2)P+r+x_2+vP+\rho+\sigma-r-1+Py_2$ for some $x_1,x_2\in[0,K)$ and $y_1,y_2\in[0,L)$. Thus, $x_1+Py_1=x_2+Py_2$ and $x_1-x_2=P(y_2-y_1)$. But $x_1-x_2\in(-K,K)$ and $P>K$, so this implies $y_1=y_2$ and $x_1=x_2$.

    \textit{Case 2:} $y$ appears in UR. Then we have $$\begin{aligned}(u-2)P+&r+x_1+vP+\rho+\sigma-r-1+Py_1=\\
    &(u-2)P+r+x_2+Py_2+z_2\end{aligned}$$
    where $x_1\in[0,K)$, $y_1\in[0,L)$, $x_2\in[0,K)$, $y_2\in[0,v)$ and $z_2\in[0,s)$. This reduces to
    $$x_1+P(v+y_1-y_2)+\rho+\sigma-1=r+x_2+z_2.$$
    Now if $y_1-y_2\ge -v+2$, then the left-hand side is at least $2P+1$, while the right-hand side is at most $r+K-1+s-1=P+s-2\le P+K-1\le 2P-1$, a contradiction. Otherwise, we must have $y_1=0$ and $y_2=v-1$, in which case $z_2\le\sigma-1$. Then the left-hand side is at least $P+\sigma$ while the right-hand side is at most $r+K-1+\sigma-1=P+\sigma-2$, a contradiction.
    
    \textit{Case 3:} $y$ appears in LL. Then we have
    $$\begin{aligned}(u-2)P&+r+x_1+vP+\rho+\sigma-r-1+Py_1=\\
    &Py_2+z_2+vP+\rho+\sigma-r-1+Pw_2\end{aligned}$$
    where $x_1\in[0,K)$, $y_1\in[0,L)$, $y_2\in[0,u)$, $z_2\in[0,r)$, and $w_2\in[0,L)$. Hence,
    $$(u-2)P+r+x_1+Py_1=Py_2+z_2+Pw_2$$
    and $x_1-z_2+r\in P\mathbb Z$. But $x_1-z_2+r\in[1,K+r-1]$, a contradiction.

    \textit{Case 4:} $y$ appears in LR. Then we have
    $$(u-2)P+r+x_1+vP+\rho+\sigma-r-1+Py_1=y_2P+z_2+y_2'P+z_2'$$
    where $x_1\in[0,K)$, $y_1\in[0,L)$, $y_2\in[0,u)$, $y_2'\in[0,v)$, $z_2\in [0,r)$, and $z_2'\in[0,s)$. This reduces to
    $$P(u-2+v+y_1-y_2-y_2')=z_2+z_2'+1-x_1-\rho-\sigma.$$ The right-hand side is at least $1-K+1-r-s\ge 2-P-s\ge 2-P-K-1>-2P$ and at most $r-1+s-1+1-1-1=r+s-3\le r+K-2<P$. Thus, $u-2+v+y_1-y_2-y_2'\in\{-1,0\}.$ But $u-2+v+y_1-y_2-y_2'\ge u-2+v+0-u+1-v+1=0$, with equality only if $y_1=0$, $y_2=u-1$, and $y_2'=v-1$. Then $z_2<\rho$ and $z_2'<\sigma$, and the right-hand side is at most $\rho-1+\sigma-1+1-0-\rho-\sigma=-1,$
    a contradiction.
\end{proof}

\begin{proof}[Proof of Theorem \ref{thm:shift-count-main-results}]
    Let $\alpha\oplus\beta$ be the degree table of $\Shift_{r,s}$. For brevity, write $a=\mathcal A$, $b=\mathcal B$, $c=\mathcal C$, and $d=\mathcal D$. We count $t(\alpha\oplus \beta)$ by counting the terms of UL, LR, UR, and LL, in that order, counting only new terms in each step. Proposition \ref{prop:shift-decodable} implies that $|t(\mathrm{UL})|=KL$, and the terms of UL are distinct from all others in the table.

    Note that $c=(P[0,u-1)+[0,r))\cup (P(u-1)+[0,\rho))$ and $d=(P[0,v-1)+[0,s))\cup(P(v-1)+[0,\sigma)).$ Thus
    $$\begin{aligned}
        c+d=&\,(P[0,u-1)+[0,r)+P[0,v-1)+[0,s))\\
        &\cup (P[0,u-1)+[0,r)+P(v-1)+[0,\sigma))\\
        &\cup(P(u-1)+[0,\rho)+P[0,v-1)+[0,s))\\
        &\cup(P(u-1)+[0,\rho)+P(v-1)+[0,\sigma))\\
        =&\,(P[0,u+v-3)+[0,r+s-1))\\
        &\cup(P[v-1,u+v-2)+[0,r+\sigma-1))\\
        &\cup(P[u-1,u+v-2)+[0,s+\rho-1))\\
        &\cup(P(u+v-2)+[0,\rho+\sigma-1))\\
        =&\,P[0,u+v-3)+[0,r+s-1)\\
        &\sqcup(P(u+v-3)+[0,\max\{r+\sigma,s+\rho\}-1))\\
        &\sqcup (P(u+v-2)+[0,\rho+\sigma-1)).
    \end{aligned}$$
    The last union is disjoint because $s\le K+1$ implies $r+s-1\le K+r=P$, and each interval appearing above has length at most $r+s-1$. Thus, $|t(\mathrm{LR})|=(u+v-3)(r+s-1)+\max\{r+\sigma,s+\rho\}+\rho+\sigma-2.$ Now
    $$\begin{aligned}
        a+d=&\,(u-2)P+r+[0,K)\\
        &+((P[0,v-1)+[0,s))\cup(P(v-1)+[0,\sigma)))\\
        =&\,P[u-2,u+v-3)+[r,K+r+s-1)\\
        &\cup( P(u+v-3)+[r,K+r+\sigma-1))\\
    \end{aligned}$$
    Now $K+r+\sigma-1\ge r+\sigma-1$ and $K+r+\sigma-1\ge s-1+r+\sigma-1\ge s-1+\rho+1-1=s+\rho-1$. Thus, the sets $P[u-2,u+v-3)+[r,r+s-1)$ and $P(u+v-3)+[r,\max\{r+\sigma,s+\rho\}-1)$ are already contained in $t(\mathrm{LR})$. We are left with
    $$\begin{aligned}
        &(P[u-2,u+v-3)+[r+s-1,K+r+s-1))\\
        \cup&(P(u+v-3)+[\max\{r+\sigma,s+\rho\}-1,K+r+\sigma-1))\\
        =&P[u-1,u+v-2)+[s-K-1,s-1)\\
        \cup&(P(u+v-2)+[\max\{\sigma-K,s+\rho-K-r\}-1,\sigma-1)).
    \end{aligned}$$
    using $P=K+r$. Now $P+s-K-1\ge r+s-1$ and $P+\max\{\sigma-K,s+\rho-K-r\}\ge\max\{\sigma+r,s+\rho\}$. It follows that the new terms in UR are
    $$\begin{aligned}
        &(P[u-1,u+v-2)+[s-K-1,0))\\
        \cup&(P(u+v-2)+[\max\{\sigma-K,s+\rho-K-r\}-1,0))
    \end{aligned}$$
    and the number of new terms contributed by UR is 
    $$ (v-1)(K+1-s)-\max\{\sigma-K,s+\rho-K-r\}+1. $$
    Now
    $$\begin{aligned}
    c+ b=&\,((P[0,u-1)+[0,r))\cup(P(u-1)+[0,\rho)))\\
    &+(vP+\rho+\sigma-r-1+P[0,L))\\
    =&\,(P[v,L+u+v-2)+[\rho+\sigma-r-1,\rho+\sigma-1))\\
    &\cup(P[u+v-1,L+u+v-1)\\
    &+[\rho+\sigma-r-1,2\rho+\sigma-r-1))\\
    =&\,(P[v,L+u+v-2)+[\rho+\sigma-r-1,\rho+\sigma-1))\\
    &\sqcup(P(L+u+v-2)\\
    &+[\rho+\sigma-r-1,2\rho+\sigma-r-1)).
    \end{aligned}$$
    Since $\rho+\sigma-r-1\ge1-r\ge1-K>-P$ and $\rho+\sigma-1\le r+s-1\le P$, the interval $[\rho+\sigma-r-1,\rho+\sigma-1)$ contains no nonzero multiple of $P$. Moreover, $2\rho+\sigma-r-1\le\rho+\sigma-1$, so neither does $[\rho+\sigma-r-1,2\rho+\sigma-r-1)$. Both intervals have length at most $r\le P$, so their translates by distinct multiples of $P$ are disjoint. Note that the terms $P[v,u+v-1)+[0,\rho+\sigma-1)$ are already contained in $c+d$. Also, all of the $$2\rho+\sigma-r-1-\rho-\sigma+r+1=\rho$$ terms in $P(L+u+v-2)+[\rho+\sigma-r-1,2\rho+\sigma-r-1)$
    and the
    $$(L+u+v-2-u-v+1)(\rho+\sigma-1-\rho-\sigma+r+1)=(L-1)r$$
    terms in $P[u+v-1,L+u+v-2)+[\rho+\sigma-r-1,\rho+\sigma-1)$ are new. It remains to analyze
    $$P[v,u+v-1)+[\rho+\sigma-r-1,0).$$
    The old terms that lie in this range are obtained by intersecting the displayed range with the union of the following five sets:
    \begin{enumerate}
        \item $P[v,u+v-2)+[\rho+\sigma-r-1,s-K-1)$
        \item $P(u+v-2)+[\rho+\sigma-r-1,\max\{r+\sigma,s+\rho\}-K-r-1)$
        \item $P(u+v-1)+[\rho+\sigma-r-1,\rho+\sigma-K-r-1)$
        \item $P[u-1,u+v-2)+[\max\{s-K-1,\rho+\sigma-r-1\},0)$
        \item $P(u+v-2)+[\max\{\sigma-K-1,s+\rho-K-r-1,\rho+\sigma-r-1\},0)$.
    \end{enumerate}
    Note that (3) is empty. Note that $\sigma-K-1=r+\sigma-K-r-1\le\max\{r+\sigma,s+\rho\}-K-r-1$ and $s+\rho-K-r-1\le\max\{r+\sigma,s+\rho\}-K-r-1$, so that the union of (2) and (5) is $$P(u+v-2)+[\rho+\sigma-r-1,0).$$ Now concerning (1) and (4), the union of $[\rho+\sigma-r-1,s-K-1)$ and $[\max\{s-K-1,\rho+\sigma-r-1\},0)$ is $[\rho+\sigma-r-1,0)$. Thus, if $v\le u-1$ then the union of (1) and (4) is
    $$\begin{aligned}
    &P[v,u-1)+[\rho+\sigma-r-1,s-K-1)\\
    \cup&P[u-1,u+v-2)+[\rho+\sigma-r-1,0)
    \end{aligned}$$
    and the remaining new elements are $P[v,u-1)+[\max\{\rho+\sigma-r-1,s-K-1\},0)$, of which there are
    $$(u-v-1)\min\{(r+1-\rho-\sigma)_+,K+1-s\};$$
    while if $v>u-1$ then the union of (1) and (4) is
    $$\begin{aligned}
        &P[u-1,v)+[\max\{s-K-1,\rho+\sigma-r-1\},0)\\
        \cup&P[v,u+v-2)+[\rho+\sigma-r-1,0).
    \end{aligned}$$
    and, after intersecting this union with the range under consideration, there are no additional new elements. Putting all this together, we find that $|t(\alpha\oplus\beta)|$ equals
    $$\begin{aligned}
       &\,KL+(u+v-3)(r+s-1)\\
        &+\max\{r+\sigma,s+\rho\}+\rho+\sigma-2\\
        &+(v-1)(K+1-s)\\
        &-\max\{\sigma-K,s+\rho-K-r\}+1\\
        &+\rho+(L-1)r\\
        &+(u-v-1)_+\min\{(r+1-\rho-\sigma)_+,K+1-s\}\\
        =&\,KL+(u+v-1)(r+s-1)-2r-2s+2\\
        &+\max\{r+\sigma,s+\rho\}+\rho+\sigma-2\\
        &+v(K-s+1)-K-1+s\\
        &-\max\{r+\sigma,s+\rho\}+K+r+1\\
        &+\rho+Lr-r\\
        &+(u-v-1)_+\min\{K-s+1,(r-\rho-\sigma+1)_+\}\\
        =&\,KL+(u+v-1)(r+s-1)-2r-s\\
        &+2\rho+\sigma\\
        &+v(K-s+1)+Lr\\
        &+(u-v-1)_+\min\{K-s+1,(r-\rho-\sigma+1)_+\}.
    \end{aligned}$$
\end{proof}

\begin{proof}[Proof of Corollary \ref{cor:shift-dog-main-results}]
Since $\rho=r$ and $\sigma=s$, the formula for $N_{\Shift_{r,s}}$ follows immediately from Theorem~\ref{thm:shift-count-main-results}. For the corresponding $\DOG_{r,s}$ construction, set $P=K+r$. Its three interference blocks are
\[
\begin{aligned}
\mathcal C+\mathcal D
&=P(L-1)+2K+P[0,u+v-1)+[0,r+s-1),\\
\mathcal C+\mathcal B
&=K+P[0,L+u-1)+[0,r),\\
\mathcal A+\mathcal D
&=P(L-1)+K+P[0,v)+[0,K+s-1).
\end{aligned}
\]
Since $r+s-1\le P$, the first block contains $(u+v-1)(r+s-1)$ terms. The first $L$ intervals of $\mathcal C+\mathcal B$ are new, while the remaining intervals are contained in $\mathcal C+\mathcal D$, so this block contributes $Lr$ additional terms.

Since $K+s-1\ge P$, the block $\mathcal A+\mathcal D$ is an interval. Its first $r$ terms are contained in $\mathcal C+\mathcal B$. The remaining new terms consist of one gap of length $K-r$ and $v-1$ gaps of length $K-s+1$ in $\mathcal C+\mathcal D$. Therefore,
\begin{multline*}
    N_{\DOG_{r,s}} 
= \\
KL+(u+v-1)(r+s-1)+Lr
+K-r+(v-1)(K-s+1).
\end{multline*}
Subtracting the formula for $N_{\Shift_{r,s}}$ gives
\[
N_{\DOG_{r,s}}-N_{\Shift_{r,s}}=s-r-1.
\]

\end{proof}

\begin{proof}[Proof of Corollary \ref{cor:balanced-shift-upper-main-results}]
    For general $n$, choose $r=\lfloor n^{1/3}\rfloor$ and $s=\lfloor n^{2/3}\rfloor$. For all sufficiently large $n$, these parameters satisfy the assumptions of Theorem \ref{thm:shift-count-main-results}. We have $r=n^{1/3}+O(1)$, $s=n^{2/3}+O(1)$, $u=n^{2/3}+O(n^{1/3})$, and $v=n^{1/3}+O(1)$, while $1\leq\rho\leq r$ and $1\leq\sigma\leq s$. Substituting these estimates into Theorem \ref{thm:shift-count-main-results} gives
    \[
    N_{\Shift_{r,s}}=n^2+3n^{4/3}+O(n).
    \]

    When $n=m^3$, choosing $r=m$ and $s=m^2$ gives $u=m^2$, $v=m$, $\rho=m$, and $\sigma=m^2$. Substituting these values into Theorem \ref{thm:shift-count-main-results} gives the stated exact formula.
\end{proof}

\begin{proof}[Proof of Proposition \ref{prop:cover-decodable}]
    The $T$-security is immediate; we prove decodability. Let $\alpha\oplus\beta$ be the degree table of $\Cover_r$, and let UL, UR, LL, and LR be quadrants of $\alpha\oplus\beta$. Suppose that $x=y$, where $x$ is a term of UL and $y$ is a term of $\alpha\oplus\beta$. 

    \textit{Case 1:} $y$ appears in UL. Then we have 
    $r+x_1+Py_1=r+x_2+Py_2$ for some $x_1,x_2\in[0,K)$ and $y_1,y_2\in[0,L)$. Thus $x_1-x_2=P(y_2-y_1)$. We have $-K<x_1-x_2<K$ and $P>K$, so this implies $y_2-y_1=0$, hence $x_1=x_2$. Then $x,y$ occupy the same position.

    \textit{Case 2:} $y$ appears in UR. Then we have $r+x_1+Py_1=r+x_2+LP+z_1$ for some $x_1,x_2\in[0,K)$, $y_1\in[0,L)$, and $z_1\in[0,T)$, so $x_1+Py_1=x_2+LP+z_1$. Thus, $K-1+P(L-1)\ge LP$, so $K-1-P\ge 0$, contradicting $P=K+r$.

    \textit{Case 3:} $y$ appears in LL. Note that $\mathcal C\subseteq P[0,\lceil T/r\rceil)+[0,r)$. Thus, we have $r+x_1+Py_1=Pz_1+w_1+Py_2$ for some $x_1\in[0,K)$, $y_1,y_2\in[0,L)$, $z_1\in[0,\lceil T/r\rceil)$, and $w_1\in[0,r)$. Hence, $r+x_1-w_1=P(z_1+y_2-y_1)$. Now, $0=r+0-r<r+x_1-w_1<r+K-0=P$, so the left-hand side is not a multiple of $P$, a contradiction.

    \textit{Case 4:} $y$ appears in LR. Then we have $r+x_1+Py_1=Pz_1+w_1+LP+u_1$ for some $x_1\in[0,K)$, $y_1\in[0,L)$, $z_1\in[0,\lceil T/r\rceil)$, $w_1\in[0,r)$, and $u_1\in[0,T).$ Thus, $r+(K-1)+P(L-1)\ge LP$, so $r+K-1-P\ge 0$, contradicting $P=K+r$.
\end{proof}

\begin{proof}[Proof of Theorem \ref{thm:cover-count-main-results}]
    Proposition \ref{prop:cover-decodable} implies that the terms of UL are distinct and distinct from all others in the table. We count the new terms contributed by LR, UR, and LL. For brevity, write $a=\mathcal A$, $b=\mathcal B$, $c=\mathcal C$, and $d=\mathcal D$. Set $u=\lceil T/r\rceil$, $\rho=T-r(u-1)$ and note that
    $c=(P[0,u-1)+[0,r))\cup(P(u-1)+[0,\rho))$. Then
    $$\begin{aligned}
    c+d=&\,((P[0,u-1)+[0,r))\cup(P(u-1)+[0,\rho)))\\
    &+(LP+[0,T))\\
    =&\,(P[0,u-1)+[0,r)+LP+[0,T))\\
    &\cup(P(u-1)+[0,\rho)+LP+[0,T))\\
    =&\,(P[L,L+u-1)+[0,T+r-1))\\
    &\cup(P(L+u-1)+[0,T+\rho-1)).
    \end{aligned}$$
        Since $K<T$, we have $T+r-1\ge K+r=P$, so the consecutive intervals above overlap or are adjacent. Their union is therefore $[PL,PL+P(u-1)+T+\rho-1)$, which contains $P(u-1)+T+\rho-1=Pu+T+\rho-P-1$ terms.
    Now
    $$\begin{aligned}
        a+d=&\,r+[0,K)+LP+[0,T)\\
        =&\,PL+[r,r+K+T-1).
    \end{aligned}$$
    Since $u\ge2$ and $\rho\ge1$, its lower endpoint is at least $PL$, while its upper endpoint satisfies $PL+r+K+T-1=PL+P+T-1\le PL+P(u-1)+T+\rho-1$. Thus, this interval is contained in $c+d$. Finally,
    $$\begin{aligned}
        c+b=&\,((P[0,u-1)+[0,r))\cup(P(u-1)+[0,\rho)))\\
        &+P[0,L)\\
        &=(P[0,L+u-2)+[0,r))\\
        &\cup(P[u-1,L+u-1)+[0,\rho))\\
        &=(P[0,L+u-2)+[0,r))\\
        &\cup(P(L+u-2)+[0,\rho)).
    \end{aligned}$$
    Every term of $c+b$ that is at least $PL$ is less than $PL+P(u-1)+T+\rho-1$, and hence already belongs to $c+d$. Therefore, the only new terms are those below $PL$, namely $P[0,L)+[0,r)$. Since $r<P$, these are $L$ disjoint intervals of length $r$, so they contribute $Lr$ terms.
    We find a total of
    $$\begin{aligned}
        &\,KL+Pu+T+\rho-P-1+Lr\\
        =&\,KL+(K+r)(u-1)+T+\rho-1+Lr\\
        =&\,KL+K(u-1)+2T+Lr-1\\
        =&\,KL+(\lceil T/r\rceil -1)K+2T+Lr-1.
    \end{aligned}$$
\end{proof}

\subsection{Proof of the General Lower Bound}

\begin{proof}[Proof of Theorem \ref{thm:general-lower-bound-main-results}]
Let $\alpha\in\mathbb Z_{\geq0}^{K+T}$ and $\beta\in\mathbb Z_{\geq0}^{L+T}$ define a decodable and $T$-secure degree table. Our strategy is to use only equality and cancellation among the entries of $\alpha\oplus\beta$ and identify a pattern of equalities forbidden by decodability.

 \begin{claim} \label{Claim forbidden submatrices}
        There do not exist indices $x,y,z,w\in[T]$, with $x\ne y$ and $z\ne w$, such that all of the following hold:
        \begin{enumerate}
            \item in the columns of UR corresponding to $\beta_{L+x}$ and $\beta_{L+y}$ there are two equal entries;
            \item in the rows of LL corresponding to $\alpha_{K+z}$ and $\alpha_{K+w}$ there are two equal entries; and
            \item $\beta_{L+x}+\alpha_{K+z}=\beta_{L+y}+\alpha_{K+w}.$
        \end{enumerate}
    \end{claim}
    \begin{proof}
        Suppose such indices $x,y,z,w$ existed. From (1) we have that $\alpha_i+\beta_{L+x}=\alpha_j+\beta_{L+y}$, for some (necessarily distinct) $i,j\in[K]$. From (2) we have $\beta_g+\alpha_{K+z}=\beta_h+\alpha_{K+w}$ for some $g,h\in [L]$. Adding these equations together, and subtracting (3), we obtain $\alpha_i+\beta_g=\alpha_j+\beta_h$. Both sides are entries of UL. Since the degree table is decodable, they must occupy the same position, so $i=j$ and $g=h$. The first two equalities then give $\beta_{L+x}=\beta_{L+y}$ and $\alpha_{K+z}=\alpha_{K+w}$. Since the masking exponents are pairwise distinct within their respective sets, this implies $x=y$ and $z=w$, contradicting the assumptions.
    \end{proof}
    We illustrate the forbidden submatrix in Figure \ref{Figure forbidden submatrix}.

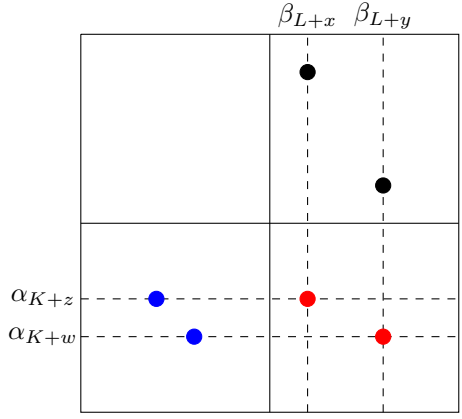
\begin{figure}
\begin{center}
\begin{tikzpicture}[scale=0.5]
\draw (0, 10) -- (10, 10) -- (10, 0) -- (0, 0) -- (0, 10);
\draw (0, 5) -- (10, 5);
\draw (5, 0) -- (5, 10);
\draw [dashed] (0, 3) -- (10, 3);
\draw [dashed] (0, 2) -- (10, 2);
\draw [dashed] (6, 0) -- (6, 10);
\draw [dashed] (8, 0) -- (8, 10);
\filldraw [black] (6, 9) circle (0.2);
\filldraw [black] (8, 6) circle (0.2);
\filldraw [red] (6, 3) circle (0.2);
\filldraw [red] (8, 2) circle (0.2);
\filldraw [blue] (2, 3) circle (0.2);
\filldraw [blue] (3, 2) circle (0.2);
\draw (6, 10.5) node{$\beta_{L+x}$};
\draw (8, 10.5) node{$\beta_{L+y}$};
\draw (-1, 3) node{$\alpha_{K+z}$};
\draw (-1, 2) node{$\alpha_{K+w}$};
\end{tikzpicture}
\end{center}
\caption{The forbidden submatrix of Claim \ref{Claim forbidden submatrices}. The two black/red/blue positions contain the same entry.} \label{Figure forbidden submatrix}
\end{figure}

Let $|t(\alpha\oplus\beta)|=KL+m$. By decodability, the $KL$ terms in UL are distinct from all other terms, so $m$ is exactly the number of distinct terms outside UL. In particular, $m\ge1$. Decodability also implies that $\alpha_1,\ldots,\alpha_K$ are pairwise distinct and that $\beta_1,\ldots,\beta_L$ are pairwise distinct, since otherwise two entries of UL in the same column or row would be equal.

The block LL contains $LT$ entries and at most $m$ distinct terms, so some term $a$ occurs at least $LT/m$ times in LL. Since $\beta_1,\ldots,\beta_L$ are pairwise distinct, $a$ occurs at most once in each row of LL. Let $I\subseteq[T]$ be the set of rows containing $a$. Then $|I|\ge LT/m$. Similarly, some term $b$ occurs in at least $KT/m$ columns of UR. Let $J\subseteq[T]$ be the set of these columns. Then $|J|\ge KT/m$.

We claim that the entries of LR in the positions in $I\times J$ are all distinct. Otherwise, suppose that two distinct positions $(i,j)$ and $(i',j')$ contain the same entry, so that $\alpha_{K+i}+\beta_{L+j}=\alpha_{K+i'}+\beta_{L+j'}$. Since the masking exponents are pairwise distinct, we must have $i\ne i'$ and $j\ne j'$. The repeated terms $a$ and $b$, together with this equality, satisfy the three conditions of Claim~\ref{Claim forbidden submatrices} with $(x,y,z,w)=(j,j',i,i')$, a contradiction. This situation is illustrated in Figure~\ref{Figure 2}.

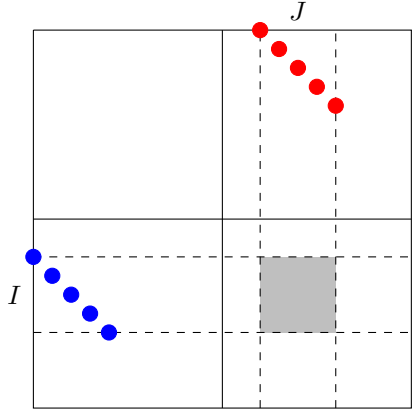
\begin{figure}
\begin{center}
\begin{tikzpicture}[scale=0.5]
\fill[lightgray] (6,2) rectangle (8,4);
\draw (0, 10) -- (10, 10) -- (10, 0) -- (0, 0) -- (0, 10);
\draw (0, 5) -- (10, 5);
\draw (5, 0) -- (5, 10);
\draw [dashed] (0, 4) -- (10, 4);
\draw [dashed] (0, 2) -- (10, 2);
\draw [dashed] (6, 0) -- (6, 10);
\draw [dashed] (8, 0) -- (8, 10);
\filldraw [red] (6, 10) circle (0.2);
\filldraw [red] (6.5, 9.5) circle (0.2);
\filldraw [red] (7, 9) circle (0.2);
\filldraw [red] (7.5, 8.5) circle (0.2);
\filldraw [red] (8, 8) circle (0.2);
\filldraw [blue] (0, 4) circle (0.2);
\filldraw [blue] (0.5, 3.5) circle (0.2);
\filldraw [blue] (1, 3) circle (0.2);
\filldraw [blue] (1.5, 2.5) circle (0.2);
\filldraw [blue] (2, 2) circle (0.2);
\draw (7, 10.5) node{$J$};
\draw (-0.5, 3) node{$I$};
\end{tikzpicture}
\end{center}
\caption{The proof of Theorem \ref{thm:general-lower-bound-main-results}. The red/blue terms are all equal, and the terms in the grey box are all distinct.}
\label{Figure 2}
\end{figure}
Thus, the number of distinct terms in LR is at least $|I\times J|\ge (LT/m)(KT/m)=KLT^2/m^2$. Since every term in LR lies outside UL by decodability, these terms are among the $m$ distinct terms counted above. Therefore, $m\ge KLT^2/m^2$, so $m\ge(KLT^2)^{1/3}$. Hence $|t(\alpha\oplus\beta)|=KL+m\ge KL+(KLT^2)^{1/3}$.
\end{proof}

    \subsection{Sharpening the lower bound in the balanced regime}
We now assume that $|{A}| = |{B}| = |{C}| = |{D}|=n$ and prove Theorem \ref{thm:balanced-lower-main-results} by sharpening the proof of Theorem \ref{thm:general-lower-bound-main-results} in this case. 


To prove Theorem \ref{thm:balanced-lower-main-results}, we combine energy arguments with a simple injectivity argument. We need the following energy estimates. For finite sets $E,F\subseteq\mathbb R$, let $r_{E+F}(x)$ denote the number of pairs $(e,f)\in E\times F$ such that $e+f=x$, and define $r_{E-F}(x)$ similarly.

\begin{lemma}\label{energy lemma 1}
    Let $A,B,C,D$ be finite sets of real numbers and $|A+B| = |A||B|$. Then 
    \[
    \sum_{x} r_{A+C}(x)r_{B+D}(x) \leq |C||D|.
    \]
\end{lemma}

\begin{proof}
    If $a+c = b+d$, then $a-b = d-c$, and hence
    \[
    \sum_{x} r_{A+C}(x)r_{B+D}(x) = \sum_x r_{A-B}(x)r_{D-C}(x).
    \]
    If $|A+B| = |A||B|$ then we also have that $|A-B| = |A||B|$ and so $r_{A-B}(x) \leq 1$. Hence $\sum_x r_{A-B}(x)r_{D-C}(x) \leq \sum_x r_{D-C}(x) = |C||D|$.
\end{proof}

\begin{lemma}\label{energy lemma 2}
    Let $A,B,C,D$ be finite nonempty sets of real numbers and $|A+B| = |A||B|$. Then
    \[
    \sum_x r_{A+C}(x)r_{C+D}(x) \leq \frac{|B+D|}{|B|}|C|^2,
    \]
    and 
    \[
    \sum_x r_{B+D}(x)r_{C+D}(x) \leq \frac{|A+C|}{|A|}|D|^2.
    \]
\end{lemma}
\begin{proof}
    If $a+c = c'+d$ then $a-d = c'-c$ and so 
    \[
    \sum_{x}r_{A+C}(x)r_{C+D}(x) = \sum_x r_{A-D}(x) r_{C-C}(x). 
    \]
    Let $x$ be such that $r_{A-D}(x)=r$ is positive. Then there are distinct $a_1,\cdots, a_r\in A$ and $d_1,\cdots, d_r \in D$ such that $d_i = a_i - x$. Note that since $|A+B| = |A||B|$, the sets $a_i + B$ are all disjoint and hence the sets $(a_i + B) - x = d_i + B$ are all disjoint. Therefore 
    \[
    |B+D| \geq \left|\bigcup_{i=1}^r(d_i + B)\right| = r_{A-D}(x) |B|.
    \]
    Thus we have 
    \begin{align*}
    &\sum_{x}r_{A+C}(x)r_{C+D}(x) =  \sum_x r_{A-D}(x) r_{C-C}(x)  \\ \leq & \frac{|B+D|}{|B|} \sum_x r_{C-C}(x) =  \frac{|B+D|}{|B|} |C|^2.
    \end{align*}
    The other inequality follows by symmetry.
\end{proof}

We are now ready to improve the lower bound.

\begin{proof}[Proof of Theorem \ref{thm:balanced-lower-main-results}]
By interchanging the names of $C$ and $D$, it is enough to prove the equivalent bound for $X=(A+C)\cup(B+D)\cup(C+D)$. If $|X|\geq3n^{4/3}$, then we are already done. Thus, we may assume that $|X|<3n^{4/3}$.

For each $x$, consider the 3 integers $r_{A+C}(x), r_{B+D}(x), r_{C+D}(x)$ and order them smallest to largest, breaking ties arbitrarily. Define $r_{min}(x), r_{mid}(x), r_{max}(x)$ to be these three values and for each $x$ label $A+C$, $B+D$ and $C+D$ with the minimum, the middle, or the maximum naturally. We partition $X$ into three sets by identifying in which of the sum sets each element of $X$ has  the maximum number of representations. That is,
\[
X = X_{A+C} \cup X_{B+D} \cup X_{C+D}
\]
where $x\in X_{A+C}$ if and only if $A+C$ was labeled with $r_{max}(x)$ and similarly for the other sets. 

We lower bound $X$ by showing that there are not many representations except those in $X_{A+C} \cup X_{B+D} \cup X_{C+D}$. To this end, let 
\[E = \sum_{x\not \in X_{A+C}} r_{A+C}(x) +    \sum_{x\not \in X_{B+D}} r_{B+D}(x) + \sum_{x\not \in X_{C+D}} r_{C+D}(x).\]

Then we have

\begin{align*}
&E  =  \sum_x \bigl(r_{min}(x) + r_{mid}(x)\bigr)\leq  \sum_x \bigl(2r_{min}(x) + r_{mid}(x)\bigr)  \\
=& \sum_{x} \min\{r_{A+C}(x),r_{B+D}(x)\} + \sum_{x} \min\{r_{B+D}(x),r_{C+D}(x)\} \\& + \sum_{x} \min\{r_{A+C}(x),r_{C+D}(x)\} \\
\leq &   \sum_x \sqrt{r_{A+C}(x)r_{B+D}(x)}+\\ &\sum_x \sqrt{r_{B+D}(x)r_{C+D}(x)}+ \sum_x \sqrt{r_{A+C}(x)r_{C+D}(x)}\\
\leq & |X|^{1/2}\left(\sum_x r_{A+C}(x)r_{B+D}(x)\right)^{1/2} + \\& |X|^{1/2}\left(\sum_x r_{B+D}(x)r_{C+D}(x)\right)^{1/2} +\\ & |X|^{1/2}\left(\sum_x r_{A+C}(x)r_{C+D}(x)\right)^{1/2} \\
\leq & |X|^{1/2}n + 2|X| n^{1/2} = O(n^{11/6}).
\end{align*}
Going from the second to the third line uses that for nonnegative numbers $a,b$ we have $\min\{a,b\} \leq \sqrt{ab}$. Going from the third to the fourth line is Cauchy-Schwarz. Going from the fourth to the fifth line uses Lemmas \ref{energy lemma 1} and \ref{energy lemma 2} and the assumption that $|A|=|B| = |C| = |D| = n$, and the last inequality uses $|X|<3n^{4/3}$. Summarizing, we have 
\begin{align}\label{duplicate weight estimate} 
    E
    \leq & |X|^{1/2}n + 2|X| n^{1/2} = O(n^{11/6}).
\end{align}

Now consider the number of quadruples $(a,b,c,d) \in A\times B\times C\times D$ such that $a+c\in X_{A+C}$, $b+d\in X_{B+D}$, and $c+d\in X_{C+D}$. Let $\mathcal{X}$ be this set of quadruples. If a quadruple does not belong to $\mathcal{X}$, then at least one of the three required conditions fails. For each failed condition, the corresponding pair is counted by one of the sums in \eqref{duplicate weight estimate}, while the remaining two coordinates can be chosen arbitrarily. Thus, the number of quadruples not satisfying the conditions is at most
\begin{align*}
& n^2\cdot E= O(n^{23/6}),
\end{align*}
and so $|\mathcal{X}| = (1-O(n^{-1/6}))n^4$. We claim that $|\mathcal{X}| \leq |X_{A+C}||X_{B+D}||X_{C+D}|$. To see this, if we have $(a,b,c,d), (a',b',c',d') \in \mathcal{X}$ with 
\begin{align*}
    a+c &= a'+c'\\
    b+d &=b'+d' \\
    c+d &=c'+d'
\end{align*}
then adding the first two equalities and subtracting the third gives $a+b=a'+b'$. Since $|A+B|=|A||B|$, this implies that $a=a'$ and $b=b'$. The first two equalities then give $c=c'$ and $d=d'$. Therefore, $(a,b,c,d)=(a',b',c',d')$. By using the AM-GM inequality, we have that 
\begin{align*}
|X| &= |X_{A+C}| + |X_{B+D}| + |X_{C+D}| \\ &\geq 
3(|X_{A+C}||X_{B+D}||X_{C+D}|)^{1/3} \\ &\geq (3 - O(n^{-1/6}))n^{4/3}.
\end{align*}
For a decodable and $n$-secure degree table with $K=L=T=n$, decodability gives $|\mathcal A|=|\mathcal B|=n$ and $|\mathcal A+\mathcal B|=n^2$, while $n$-security gives $|\mathcal C|=|\mathcal D|=n$. Thus, the theorem applies to the interference union $(\mathcal A+\mathcal D)\cup(\mathcal C+\mathcal B)\cup(\mathcal C+\mathcal D)$. Decodability also ensures that the $n^2$ useful terms in $\mathcal A+\mathcal B$ are disjoint from this interference union. Therefore, $\Nstar(n,n,n)\geq n^2+3n^{4/3}-O(n^{7/6})$.
\end{proof}

\bibliographystyle{IEEEtran}
\bibliography{references}

\end{document}